%% file: main.tex
\documentclass[11pt]{article}

\usepackage[english]{babel}

\usepackage[letterpaper,top=2cm,bottom=2cm,left=3cm,right=3cm,marginparwidth=1.75cm]{geometry}
\usepackage[utf8]{inputenc}

\usepackage{fullpage}
\usepackage{hyperref,verbatim, enumerate, amsmath,amssymb,amsfonts,mathrsfs, amsthm, mathtools, bbm, algorithm, algpseudocode}
\usepackage[capitalize]{cleveref}
\usepackage{comment}
\usepackage{fullpage}
\usepackage{xcolor}
\usepackage{lipsum}
\usepackage[font=itshape]{quoting}
\usepackage[bottom]{footmisc}
\usepackage{csquotes}
\usepackage{todonotes}
\usepackage[normalem]{ulem}

\newcommand{\eps}{\epsilon}

\newcommand{\lmi}{\mathsf{Linear \ Matroid \ Intersection}}

\newcommand{\emphdef}[1]{{\sf {#1}}}

\newcommand\restr[2]{{
  \left.\kern-\nulldelimiterspace 
  #1 
  \vphantom{\big|} 
  \right|_{#2} 
  }}
\allowdisplaybreaks[4]

\input{macros}

\newif\ifblind

\blindtrue

\title{Linear Matroid Intersection is in Catalytic Logspace}

\author{}
\ifblind
\author{Aryan Agarwala\\Max-Planck-Institut f\"{u}r Informatik \\ \texttt{aryan@agarwalas.in} \and Yaroslav Alekseev\thanks{Supported by ISF grant 507/24.} \\ Technion Israel Institute of Technology \\ \texttt{tolstreg@gmail.com}  \and
Antoine Vinciguerra\footnotemark[1] \\ Technion Israel Institute of Technology \\ \texttt{antoine.v@campus.technion.ac.il}}

\fi
\date{}

\begin{document}

\setlength{\abovedisplayskip}{5pt}
\setlength{\belowdisplayskip}{5pt}

\pagenumbering{gobble}

\maketitle

\input{structure_paper/abstract}


\pagebreak

\pagenumbering{arabic}

\input{structure_paper/introduction}

\input{structure_paper/prelims}

\input{structure_paper/unique_minimum_weight_matching}
\input{structure_paper/maximal_independent_set}
\input{structure_paper/from_k_to_k_p_1}
\input{structure_paper/final_algorithm}
\input{structure_paper/conclusion}

\ifblind
\section*{Acknowledgements}
All three authors thank Ian Mertz and Yuval Filmus for extensive discussions.
\fi


\DeclareUrlCommand{\Doi}{\urlstyle{sf}}
\renewcommand{\path}[1]{\small\Doi{#1}}
\renewcommand{\url}[1]{\href{#1}{\small\Doi{#1}}}
\bibliographystyle{alphaurl}
\newpage

\bibliography{bibliography}
\newpage 
\appendix
\input{structure_paper/appendix}






\end{document}

%% file: macros.tex
\usepackage{blindtext}
\usepackage{enumitem}
\usepackage{caption}
\usepackage{placeins}
\usepackage{graphicx}

\usepackage{fancybox}
\usepackage{hyperref}

\usepackage{amsmath,amsfonts,amsthm,amssymb,xcolor}
\usepackage{bm} 
\usepackage{bbm}

\usepackage{nicefrac}

\usepackage{float}

\newtheorem{theorem}{Theorem}[section]
\newtheorem{corollary}[theorem]{Corollary}
\newtheorem{lemma}[theorem]{Lemma}
\newtheorem{observation}[theorem]{Observation}

\newtheorem*{theorem*}{Theorem}
\newtheorem*{corollary*}{Corollary}
\newtheorem*{conjecture*}{Conjecture}
\newtheorem*{lemma*}{Lemma}
\newtheorem*{thm*}{Theorem}
\newtheorem*{prop*}{Proposition}
\newtheorem*{obs*}{Observation}

\newtheorem*{remark*}{Remark}
\newtheorem*{rec*}{Recommendation}

\newtheorem{definition}[theorem]{Definition}

\newtheorem*{definition*}{Definition}

\newenvironment{fminipage}%
  {\begin{Sbox}\begin{minipage}}%
  {\end{minipage}\end{Sbox}\fbox{\TheSbox}}

\newcommand{\poly}{\mathrm{poly}}

\DeclareMathOperator*{\argmin}{arg\,min}

\DeclareMathOperator*{\rank}{rank}

\newcommand{\newclass}[2]{\newcommand{#1}{{\text{\upshape\sffamily #2}}\xspace}}
\renewcommand{\P}{{\text{\upshape\sffamily P}}\xspace}
\newclass{\NP}{NP}
\newclass{\ZPP}{ZPP}
\newclass{\coNP}{coNP}
\newclass{\BPP}{BPP}
\newclass{\Logspace}{L}
\newclass{\NL}{NL}
\newclass{\coNL}{coNL}
\newclass{\UL}{UL}
\newclass{\coUL}{coUL}
\newclass{\BPL}{BPL}
\newclass{\PL}{PL}
\newclass{\prBPL}{prBPL}
\newclass{\PSPACE}{PSPACE}
\newclass{\EXP}{EXP}
\newclass{\EXPSPACE}{EXPSPACE}
\newclass{\TIME}{TIME}
\newclass{\SPACE}{SPACE}
\newclass{\NSPACE}{NSPACE}
\newclass{\SC}{SC}
\newclass{\coNSPACE}{coNSPACE}
\newclass{\BPSPACE}{BPSPACE}
\newclass{\TFNP}{TFNP}

\newclass{\NC}{NC}
\newclass{\NCo}{NC$^1$}
\newclass{\ACz}{AC$^0$}
\newclass{\ACo}{AC$^1$}
\newclass{\SACo}{SAC$^1$}
\newclass{\TC}{TC}
\newclass{\TCz}{TC$^0$}
\newclass{\TCo}{TC$^1$}
\newclass{\NCt}{NC$^2$}
\newclass{\RNC}{RNC}
\newclass{\PSDNC}{pseudo-deterministic NC}
\newclass{\NUSPL}{non-uniform SPL}
\newclass{\RNCt}{RNC$^2$}
\newclass{\RNCtt}{RNC$^3$}
\newclass{\RNCo}{RNC$^1$}
\newclass{\QNC}{Quasi-NC}
\newclass{\VP}{VP}
\newclass{\CL}{CL}
\newclass{\CLP}{CLP}
\newclass{\CSPACE}{CSPACE}

\newclass{\GapL}{GapL}

\newclass{\MATCH}{MATCH}
\newclass{\DET}{DET}
\newclass{\LOSSY}{LOSSY}
\newclass{\LOSSYNC}{LOSSY[NC]}
\newclass{\LOSSYC}{LOSSY[$\mathcal{C}$]}

\newclass{\ZPC}{ZP-$\mathcal{C}$}
\newclass{\ZPNC}{ZPNC}
\newclass{\Comp}{Comp}
\newclass{\Decomp}{Decomp}
\newcommand{\I}{\mathbb{I}}
\newcommand{\Ind}[1]{\I\left[#1\right]}

%% file: structure_paper/abstract.tex
\begin{abstract}
\noindent
Linear matroid intersection is an important problem in combinatorial optimization. Given two linear matroids over the same ground set, the linear matroid intersection problem asks you to find a common independent set of maximum size. The deep interest in linear matroid intersection is due to the fact that it generalises many classical problems in theoretical computer science, such as bipartite matching, edge disjoint spanning trees, rainbow spanning tree, and many more. \\

\noindent
We study this problem in the model of catalytic computation: space-bounded machines are granted access to \textit{catalytic space},
which is additional working memory that is full with arbitrary data that
must be preserved at the end of its computation. \\

\noindent
Although linear matroid intersection has had a polynomial time algorithm for over 50 years, it remains an important open problem to show that linear matroid intersection belongs to any well studied subclass of $\P$. We address this problem for the class catalytic logspace ($\CL$) with a polynomial time bound ($\CLP$). \\

\noindent
Recently, Agarwala and Mertz (2025) showed that bipartite maximum matching can be computed in the class $\CLP\subseteq \P$. This was the first subclass of $\P$ shown to contain bipartite matching, and additionally the first problem outside $\TC^1$ shown to be contained in $\CL$. We significantly improve the result of Agarwala and Mertz by showing that linear matroid intersection can be computed in $\CLP$.
\end{abstract}

%% file: structure_paper/introduction.tex
\section{Introduction}
\subsection{Catalytic Computing}

Catalytic computation was introduced by Buhrman et al. ~\cite{BuhrmanCleveKouckyLoffSpeelman14} in order to study the power of used space. In this model, a space-bounded Turing machine is augmented with an additional read-write tape, known as the \textit{catalytic tape}. The catalytic tape is initialized adversarially with some arbitrary content $\tau$. The Turing machine may use this tape freely, with the requirement that upon termination the catalytic tape must be reset to its original state $\tau$.\\

\noindent
$\CL$ is the class of problems that can be solved by a catalytic machine with a work tape of size $O(\log n)$ and a catalytic tape of size $\poly(n)$. $\CLP$ is the class formed by the additional restriction that the machine must run in polynomial time. \\

\noindent
Although it was earlier informally conjectured~\cite{CookMckenzieWehrBravermanSanthanam12} that used space could not provide additional computational power, Buhrman et al.~\cite{BuhrmanCleveKouckyLoffSpeelman14} showed the surprising result that $\CLP$ is likely much stronger than $\Logspace$: \\
$$\Logspace \subseteq \NL \subseteq \TCo \subseteq \CLP \subseteq \CL \subseteq \mathsf{LOSSY} \subseteq  \ZPP$$

\noindent
Following the work of~\cite{BuhrmanCleveKouckyLoffSpeelman14}, catalytic computation has been a subject of growing interest, and many variants of the model have been studied, including non-deterministic and randomized ~\cite{BuhrmanKouckyLoffSpeelman18,DattaGuptaJainSharmaTewari20,CookLiMertzPyne25,KouckyMertzPyneSami25}, non-uniform~\cite{Potechin17,RobereZuiddam21,CookMertz22,CookMertz24}, error-prone~\cite{GuptaJainSharmaTewari24,FolkertsmaMertzSpeelmanTupker25}, communication~\cite{PyneSheffieldWang25}, and many more~\cite{GuptaJainSharmaTewari19,BisoyiDineshSarma22,BisoyiDineshRaiSarma24, BuhrmanFolkertsmaMertzSpeelmanStrelchukSubramanianTupker2025} (see surveys by Kouck\'{y}~\cite{Koucky16} and Mertz~\cite{Mertz23}). This interest in catalytic computation culminated in space efficient tree evaluation algorithms by Cook and Mertz~\cite{CookMertz20, CookMertz21, CookMertz22, CookMertz24}, which recently led to the breakthrough result $\TIME(t)\subseteq\SPACE(\sqrt{t\log{t}})$ by Ryan Williams~\cite{Williams25}.\\

\noindent
Despite this long line of work, however, the exact strength of catalytic computation remains unclear. Of particular interest is the relationship between $\CL$ and the $\NC$ hierarchy. Buhrman et al.~\cite{BuhrmanCleveKouckyLoffSpeelman14} showed that $\TCo \subseteq \CL$, so the two natural questions which follow are:
\begin{enumerate}
    \item Is $\NC^2 \subseteq \CL$? That is, can the $\TCo$ inclusion of \cite{BuhrmanCleveKouckyLoffSpeelman14} be strengthened?
    \item Is $\CL \subseteq \NC$? That is, can $\CL$ shown to be contained in the $\NC$ (or equivalently $\TC$) hierarchy? 
\end{enumerate}
On the first problem, Alekseev et al.~\cite{AlekseevFilmusMertzSmalVinciguerra2025} recently made progress by showing that $\mathsf{SAC^2}$ can be solved with $O(\log^2n/\log\log n)$ free space and $2^{O(\log^{1 + \eps}n)}$ catalytic space. On the second problem, Agarwala and Mertz~\cite{AgarwalaMertz25} recently presented a barrier by showing that bipartite matching, which is currently incomparable to the $\NC$ hierarchy, is contained in $\CLP$. This was the first new problem shown to lie in $\CL$ since the decade old result $\TCo \subseteq \CLP$~\cite{BuhrmanCleveKouckyLoffSpeelman14}. A natural open problem posed in~\cite{AgarwalaMertz25} is to extend their framework to solve harder problems in $\CLP$. One such problem is linear matroid intersection. 


\subsection{Linear Matroid Intersection}
A \emph{matroid}, defined by Whitney~\cite{Whitney35}, is a set-independence structure which naturally arises in many combinatorial optimization problems. Formally, a matroid is a pair $M=(S,\mathcal{I})$, where $S$ is some finite set and $\mathcal{I}\subseteq 2^S$ is a collection of subsets of $S$ called independent sets. The independent sets are required to satisfy three properties: the empty set is independent, the independent sets are downward closed, and the augmentation property. See the preliminaries for a formal definition. \\

\noindent
In the \emph{matroid intersection} problem, one is given two matroids over the same ground set, say $M_1 = (S, \mathcal{I}_1)$ and $M_2 = (S, \mathcal{I}_2)$. The goal is to find $I \in \mathcal{I}_1 \cap \mathcal{I}_2$ of maximum size. This problem is inherently challenging because while $M_1$ and $M_2$ are matroids, their intersection $(S, \mathcal{I}_1 \cap \mathcal{I}_2)$ may not be. In this paper we work exclusively with a well studied class of matroids known as \emph{linear matroids}. \\

\noindent
A linear matroid $M = (S, \mathcal{I})$ is one where the elements are a subset of a vector space, i.e, $S \subseteq \mathbb{F}^m$, and the independent sets $\mathcal{I}$ are exactly those subsets of $S$ which are linearly independent over $\mathbb{F}^m$. The $\lmi$ problem is matroid intersection where the input matroids $M_1$ and $M_2$ are linear matroids.\\

\noindent
Many important problems in combinatorial optimization are special cases of linear matroid intersection. For example:
\begin{itemize}
    \item 
    \textbf{Bipartite maximum matching:} Given a bipartite graph $G$, output a matching of maximum size.
    \item
    \textbf{Rainbow spanning tree:} Given an edge coloured graph $G$, output a spanning tree consisting of distinctly coloured edges. 
    \item 
    \textbf{Edge disjoint spanning trees:} Given a graph $G$, output two edge disjoint spanning trees.
\end{itemize}

\noindent
Bipartite matching, in particular, is closely related to linear matroid intersection. As mentioned above, bipartite matching is a special case of linear matroid intersection. On the other hand, algorithms for bipartite matching tend to influence algorithms for linear matroid intersection. The augmenting paths framework for bipartite matching~\cite{Kuhn55} led to polynomial time algorithms for linear matroid intersection~\cite{AignerDowling1971, Lawler1975}, the isolation lemma framework for bipartite matching~\cite{MulmuleyVaziraniVazirani87} led to an $\RNC$ algorithm for linear matroid intersection~\cite{NarayananSaranVazirani94}, and the $\QNC$ algorithm for bipartite matching~\cite{FennerGurjarThierauf16} led to a $\QNC$ algorithm for linear matroid intersection~\cite{GurjarThierauf17}. \\

\noindent
Recently, Agarwala and Mertz showed that bipartite maximum matching is in $\CLP$~\cite{AgarwalaMertz25}. It is then a natural question to ask whether their techniques can be extended to work for linear matroid intersection.

\subsection{Isolation Lemma}
\noindent
An key part of this paper is the celebrated \emph{isolation lemma}. Let $S$ be a ground set and $\mathcal{I} \subseteq 2^S$ be \textit{any} collection of subsets of $S$. The isolation lemma states that if one assigns polynomially-bounded integer weights uniformly and independently at random to each element of $S$, then the minimum weight set $I \in \mathcal{I}$ will be unique with high probability.\\

\noindent
Introduced by Mulmuley, Vazirani, and Vazirani in~\cite{MulmuleyVaziraniVazirani87}, the isolation lemma has since become a key tool in the design of randomised algorithms for many classical problems~\cite{OrlinStein93,LingasPersson15,NarayananSaranVazirani94,GurjarThierauf17,KlivansSpielman01,AllenderReinhardt00,BourkeTewariVinodchandran09,KalampallyTewari16,MelkebeekPrakriya19, ArvindMukhopadhyay08, AgarwalaMertz25, AnariVazirani19}. Thus, derandomizing the isolation lemma has in itself become an important problem~\cite{ChariRohatiSrinivasan93,ArvindMukhopadhyay08,AgarwalGurjarThierauf20,GurjarThieraufVishnoi21}.\\

\noindent
The work of Narayanan et al.~\cite{NarayananSaranVazirani94} used the isolation lemma to obtain an $\RNC$ algorithm for linear matroid intersection. In particular, given input matroids $M_1 = (S, \mathcal{I}_1)$ and $M_2 = (S, \mathcal{I}_2)$, they need polynomially bounded weights $w:S \rightarrow \mathbb{Z}$ such that the minimum weight maximum sized common independent set $I \in \mathcal{I}_1 \cap \mathcal{I}_2$ is unique. Gurjar and Thierauf~\cite{GurjarThierauf17} partially derandomized the isolation lemma for linear matroid intersection in $\QNC$, but obtain a weight assignment which has large quasi-polynomially (instead of polynomially) bounded weights. It is a big open problem to fully derandomize the isolation lemma for linear matroid intersection in $\NC$. \\

\noindent
Agarwala and Mertz~\cite{AgarwalaMertz25} made progress on this problem by providing a derandomization for the case of bipartite matching in $\CL$. It is then a natural goal to obtain a similar $\CL$ derandomization of the isolation lemma for linear matroid intersection. 

\subsection{Our Results}
In this paper we prove the following:
\begin{theorem}
\label{thm:lmi-in-clp}
$$\lmi \in \CLP$$
\end{theorem}

\noindent
This result is interesting for two reasons:
\begin{enumerate}
    \item $\lmi$ is now the hardest problem known to be solvable in $\CL$. The previous strongest inclusion was bipartite matching~\cite{AgarwalaMertz25}, which is a special case of linear matroid intersection. Thus, our result constitutes a stronger barrier against the $\CL \subseteq \NC$ conjecture~\cite{Koucky16, Mertz23}. 
    
    \item This is the first algorithm for $\lmi$ which uses sublinear free space and polynomial time with access to any additional resources, such as randomness, non-determinism, or, in our case, catalytic space. As far as we are aware, the only other sublinear space algorithm runs in $O(\log^2n)$ space but not polynomial time, as a corollary of the fact that $\lmi \in \QNC^2$~\cite{GurjarThierauf17}. 
\end{enumerate}

\noindent
Moreover, our algorithm, which is a natural extension of the algorithm of~\cite{AgarwalaMertz25}, derandomizes the isolation lemma for linear matroid intersection in $\CLP$.

\subsection{Proof Overview}

\noindent
We present here a high-level overview of our proof. Our proof structure is largely inspired by the techniques used to prove that maximum bipartite matching is in $\CLP$ \cite{AgarwalaMertz25}, with extra machinery needed to handle the more complicated structure of matroid intersection.\\

\noindent
The core idea of the proof is to construct an isolating weight assignment on the catalytic tape, and then apply the algorithm of of~\cite{NarayananSaranVazirani94} to compute a maximum size common independent set. Let $M_1 = (S, \mathcal{I}_1)$ and $M_2 = (S, \mathcal{I}_2)$ be the input linear matroids. We start by dividing the catalytic tape into three sections:
\begin{enumerate}
    \item 
    A weight assignment $W:S \rightarrow \mathbb{Z}$,
    \item 
    A set of 'reserve weights' used to modify $W$, and
    \item 
    A section used as catalytic space for the computation of catalytic subroutines.
\end{enumerate}

\noindent
Our algorithm proceeds iteratively with a counter $k$, starting with $k = 0$. At each step, we maintain the invariant that the current weight assignment $W$ induces a unique minimum weight (or \textit{isolates} a) size $k$ common independent set $I_k \in \mathcal{I}_1 \cap \mathcal{I}_2$. We then perform the following steps: 
\begin{enumerate}
    \item 
    First, we check if $I_k$ is a common independent set of maximum size. If it is, we output $I_k$ as our solution.
    \item 
    If $I_k$ is not a maximum size common independent, we check if $W$ also isolates a size $k+1$ common independent set $I_{k+1}$.  If it does, we increment $k$ by one and repeat the process from step $1$.
    \item 
    If $W$ does not isolate a size $k+1$ common independent set, we must start again with a new weight assignment. This is the crucial step. We swap the weight of a special element $s$, in the first section, with an arbitrary weight from the reserve section. We then use a novel compression-decompression algorithm to compress the reserve section. We reset our counter $k$ to $0$ and begin the process again with the modified weights.
\end{enumerate}

\noindent
This iterative process stops when either:
\begin{enumerate}
    \item Step 1 is successful and we find a maximum sized common independent set through a good weight assignment $W'$ on the catalytic tape, or
    \item We visit step 3 $\poly(n)$ many times, at which point we have freed up $\poly(n)$ space on the catalytic tape through compression, and we may use this space to run a standard polynomial-time algorithm for linear matroid intersection.
\end{enumerate}

\noindent
We can then use our decompression algorithm to revert the catalytic tape to its original state. This is an implementation of the compress-or-random framework introduced by Cook et al.~\cite{CookLiMertzPyne25}.\\

\noindent
Our main contribution is a novel compression and decompression scheme. Note that if our weight assignment $W$ does not isolate a size $k+1$ common independent set, then there must exist at least two size $k+1$ common independent set $I$ and $I'$. We can thus find an element $s \in I \setminus I'$. We call $s$ a \textit{threshold element}. \footnote{The concept of threshold elements was introduced by Mulmuley, Vazirani, and Vazirani~\cite{MulmuleyVaziraniVazirani87} in their proof of the isolation lemma.} Agarwala and Mertz~\cite{AgarwalaMertz25} observe that, in the context of bipartite matching, the weight of a threshold element can be deleted and later reconstructed given $s$ and $k$. However, linear matroid intersection differs from bipartite matching in two ways:
\begin{enumerate}
    \item In the case of bipartite matching, one can always ensure that $s \notin I_k$ ($s$ is not in the isolated size $k$ matching). In the case of linear matroids, this is not always possible - all threshold elements may be in $I_k$. 

    \item Agarwala and Mertz~\cite{AgarwalaMertz25} use a structure known as the \textit{residual graph} in their compression and decompression procedure. The key property they use is that shortest paths in the residual graph are in bijection with minimum weight size $k+1$ matchings of the original graph. We use a similar structure, known as the \textit{exchange graph}, for matroids. However, this bijection property no longer holds.
\end{enumerate}

\noindent
We handle both of these issues using \textit{inclusion} and \textit{exclusion} matroids. In particular, in order to execute compression and decompression using a threshold element $s$, we need to answer two questions:
\begin{enumerate}
    \item What is the minimum weight of a size $k+1$ common independent set containing $s$ (excluding the weight of $s$ itself)? These independent sets are characterised by the inclusion matroid.
    \item What is the minimum weight of a $k+1$ common independent set forbidden from containing $s$? These independent sets are characterised by the exclusion matroid. 
\end{enumerate}

\noindent
Agarwala and Mertz~\cite{AgarwalaMertz25} showed that, given a threshold element $s$, and both of the aforementioned values, one can recover the weight of $s$ as the difference of the first and second value. Thus, our main contribution is a $\CLP$ algorithm which computes both of these values.

\subsection{Organization of the Paper}

Our paper is divided into five main sections:.\\

\noindent
In \Cref{prelims} we formally introduce catalytic classes and matroids, and describe some generic catalytic subroutines that we will use later. In \Cref{isolation_lemma}, we present a $\CLP$ algorithm which, given access to a weight function $W$ which isolates a size $k$ common independent set, constructs and outputs the isolated set. This is largely based on the algorithm in~\cite{NarayananSaranVazirani94}. In \Cref{maximal_independent_set}, we present a $\CLP$ algorithm which decides whether a common independent set $I_k$ is of maximum size. In \Cref{section_intersecting_set_from_k_to_k_+_1} we present a $\CLP$ algorithm which either certifies that $W$ isolates a size $k+1$ common independent set, or identifies a threshold element. We need these algorithms in the case where $I_k$ is not of maximum size. In \Cref{final_algorithm}, we present our compression and decompression procedures, and describe the final $\CLP$ algorithm for $\lmi$.

%% file: structure_paper/prelims.tex
\section{Preliminaries}
\label{prelims}
We use notation $\mathbb{Z}^{\leq c}$ to denote the non-negative integers of value at most $c$. For $n \in \mathbb{N}$, we use $[n]$ to denote the set $\{1, \dots, n\}$. 

\noindent
Let $G(V,E)$ be a graph, for any walk $P$ of $G$, we define the hop-length of $P$ to be the number of edges in $P$. This is in order to distinguish from the weight of a walk when we work with weighted graphs.
\subsection{Catalytic Computation}
Our main computational model in this paper is the catalytic space model:
\begin{definition}[Catalytic machines]
    Let $s := s(n)$ and $c := c(n)$. A \emphdef{catalytic Turing machine} with
    space $s$ and \emphdef{catalytic space} $c$ is a Turing machine $M$ with a
    read-only input tape of length $n$, a write-only output tape, a read-write work tape of length $s$,
    and a second read-write work tape of length $c$ called the \emphdef{catalytic tape},
    which will be initialized to an adversarial string $\tau$. \\

    \noindent
    We say that $M$ computes a function $f$ if for every $x \in \{0,1\}^n$ and
    $\tau \in \{0,1\}^c$, the result of executing $M$ on input $x$ with initial
    catalytic tape $\tau$ fulfils two properties: 1) $M$ halts with $f(x)$ written on the output tape; and
    2) $M$ halts with the catalytic tape in state $\tau$.
\end{definition}

\noindent
Such machines naturally give rise to complexity classes of interest:

\begin{definition}[Catalytic classes]
    We define $\CSPACE[s,c]$ to be the family of functions computable by
    catalytic Turing machines with space $s$ and catalytic space $c$.
    We also define \emphdef{catalytic logspace} as
    $$\CL := \bigcup_{d \in \mathbb{N}} \CSPACE[d \log n, n^d]$$
    Furthermore we define $\CLP$ as
    the family of functions computable by $\CL$ machines that are
    additionally restricted to run in polynomial time for every
    initial catalytic tape $\tau$.
\end{definition}

\noindent
Important to this work will be the fact, due to Buhrman et al.~\cite{BuhrmanCleveKouckyLoffSpeelman14},
that $\CLP$ can simulate log-depth threshold circuits:

\begin{theorem}[\cite{BuhrmanCleveKouckyLoffSpeelman14}]
\label{thm:bckls}
    $$\TCo \subseteq \CLP$$
\end{theorem}

\noindent
The algorithm we present in this paper is in $\CLP$. We do not expicitly argue this due to the following theorem from Cook et al.~\cite{CookLiMertzPyne25}, that any problem solvable independently in $\CL$ and in $\P$, can be solved in $\CLP$. 
\begin{theorem}[\cite{CookLiMertzPyne25}]
\label{thm:clmp}
  $$\CLP= \CL\cap \P$$
\end{theorem}

\noindent
Due to the fact that linear matroid intersection is known to be in $\P$~\cite{Edmonds2009, Edmonds2003, Edmonds1979}, the main goal of this paper is to prove that there is a $\CL$ algorithm for Linear Matroid Intersection:
\begin{theorem}
\label{thm:lmi-in-cl}
    $$\lmi \in \CL$$
\end{theorem}

\noindent
We thus obtain \Cref{thm:lmi-in-clp} as a corollary of \Cref{thm:lmi-in-cl} and \Cref{thm:clmp}.

\subsection{Matroids}
We denote by $S\subseteq [n]$ a finite set, where $n\in \mathbb{N}$.

\begin{definition}[Matroid]
A matroid $M$ is a tuple $(S,\mathcal{I})$, where $S$ is called the ground set, and $\mathcal{I}\subseteq 2^S$ is a collection of subsets of $S$, known as 'independent sets'. The following properties must hold for $(S, \mathcal{I})$ to be a matroid:
\begin{itemize}
    \item 
    $\emptyset\in \mathcal{I}$. The empty set is independent.
    \item
    If $A\in \mathcal{I}$ and $B\subseteq A$, then $B\in \mathcal{I}$. The independent sets are downward closed. 
    \item 
    If $A,B\in \mathcal{I}$ where $\lvert A\rvert> \lvert B\rvert$, then there exists $x\in A \setminus B$ such that $B\cup \{x\}\in I$.
 \end{itemize}  
The inclusion-wise maximal sets of $\mathcal{I}$ are referred to as the `bases' of $M$. The rank of $M$ is defined to be the size of the largest independent set in $\mathcal{I}$:
\begin{equation*}
    \rank(M)\colon = \max_{I\in\mathcal{I}} \lvert I\rvert.
\end{equation*}
\end{definition}
\ \\
\noindent
In the rest of the paper, we will consider \textit{weighted} matroids.\\

\noindent
A weighted matroid is a matroid $M=(S,\mathcal{I})$, where the elements of $S$ are given integer weights $W\colon S \rightarrow \mathbb{Z}$. The weight of an independent set $I \subseteq S$ is defined to be $W(I) = \sum_{s \in I} W(s)$.\\

\begin{definition}[Common Independent Set]
Let $M_1 = (S, \mathcal{I}_1)$ and $M_2 = (S, \mathcal{I}_2)$ be weighted matroids with weights $W\colon S \rightarrow \mathbb{Z}$. \\

\noindent
$I \subseteq S$ is defined to be a `common independent set' of $M_1$ and $M_2$ if $I \in \mathcal{I}_1 \cap \mathcal{I}_2$. \\

\noindent
Additionally, we define $\min_k(M_1, M_2) = \min \{w(I) \ | \ I \in \mathcal{I}_1 \cap \mathcal{I}_2, |I| = k\}$ to be the minimum weight of a common independent set of size $k$, and  $\mathcal{I}^k_{\min}(M_1, M_2) = \{I \in \mathcal{I}_1 \cap \mathcal{I}_2 \ | \ w(I) = \min_k(M_1, M_2) \}$ to be the set of minimum weight size $k$ common independent sets.
\end{definition}

\noindent
Naively, a matroid may have an exponential (in $|S|$) number of independent sets. A fundamental challenge of formalizing computational tasks on matroids is in finding a succinct description of the independent sets. In this work, we study a large class of matroids known as \textit{linear matroids}, which can be represented succinctly by matrices. 

\begin{definition}[Linear Matroids]
\noindent
Let $M = (S = \{s_1, \dots, s_n\}, \mathcal{I})$ be a matroid, and $A$ be a matrix of dimensions $m \times n$ over a field $\mathbb{F}$. For $i \in [n]$, let $A_i \in \mathbb{F}^m$ refer to the $i^{th}$ column of the matrix $A$. $A$ is a linear representation of the matroid $M$ if:
\[\forall I \subseteq [n], \ \{s_i\mid i \in I\} \in \mathcal{I} \iff \{A_i \mid i \in I\} \text{ is linearly independent over }\mathbb{F}^m \]

\noindent
A matroid $M$ is defined to be a linear matroid if it can be linearly represented by a matrix $A$. For notational convenience, when no confusion arises, we will often refer to the matroid $M$ by its corresponding matrix.
\end{definition}

\begin{definition}[$\lmi$]
The $\lmi$ problem takes as input two linear matroids $M_1 = (S, \mathcal{I}_1)$ and $M_2 = (S, \mathcal{I}_2)$ in the form of their linear representations, and outputs a maximum sized common independent set $I$ of $M_1$ and $M_2$.
\end{definition}

\noindent
Let us now define for every matroid $M = (S, \mathcal{I})$, and every element $s \in S$, two associated matroids.
\begin{definition}[Exclusion and Inclusion Matroids]
Let $M = (S, \mathcal{I}), \ w\colon S \rightarrow \mathbb{Z}$ be a weighted matroid, and let $s \in S$.\\

\noindent
The exclusion matroid $M - s$ is defined as $M - s=(S \setminus \{s\}, \{I \in \mathcal{I} \ | \ s \notin I)$ with weights $w \restriction_{S \setminus \{s\}}$. \\

\noindent
The inclusion matroid $M\restriction_s$ is defined as $M\restriction_s=(S \setminus \{s\}, \ \{I \setminus \{s\} \mid I \in \mathcal{I}, \ s \in I\})$ with weights $w \restriction_{S \setminus \{s\}}$. \footnote{The exclusion matroid is always a matroid by definition. The inclusion matroid is a matroid if and only if $s$ is not a loop. That is, $\{s\} \in \mathcal{I}$. We will assume, without loss of generality, that this is always the case.}
\end{definition}

\begin{definition}[Membership Oracle]
\noindent
For a matroid $M = (S, \mathcal{I})$, a membership oracle $\mathcal{O}$ takes as input $I \subseteq S$ and accepts if and only if $I \in \mathcal{I}$.
\end{definition}

\noindent
A key property of linear matroids, and thus inclusion/exclusion matroids built from linear matroids, is that membership can be decided in $\CL$.
\begin{lemma}
\label{linear-matroid-membership-testing-is-in-cl}
    Given a linear matroid $M = (S, \mathcal{I})$ and $I \subseteq S$, there exists a $\CL$ algorithm which tests whether $I \in \mathcal{I}$.
\end{lemma}
\begin{proof}
    Testing whether $I \in \mathcal{I}$ involves only testing whether the linear representation of $M$, when restriction to the columns of $I$, has rank $|I|$. This can be done in $\TCo$ and thus $\CL$~\cite{AllenderBealsOgihara96, BuhrmanCleveKouckyLoffSpeelman14}.
\end{proof}

\begin{lemma}
\label{exclusion-inclusion-membership-testing-is-in-cl}
Given a linear matroid $M = (S, \mathcal{I})$, $s \in S$, and $I \subseteq S \setminus \{s\}$, there exist $\CL$ algorithms to test whether $I$ is independent for both the exclusion matroid $M - s$ and the inclusion matroid $M\restriction_s$.
\end{lemma}
\begin{proof}
For the exclusion matroid, $I$ is independent if and only if $I$ is independent in $M$. For the inclusion matroid, $I$ is independent if and only if $I \cup \{s\}$ is independent in $M$. Both can be tested in $\CL$ using \Cref{linear-matroid-membership-testing-is-in-cl}.
\end{proof}

\noindent
From now on, unless explicitly stated otherwise, all matroids in this paper will be either linear matroids, or inclusion/exclusion matroids built from linear matroids.

\subsection{Graph Algorithms}
 We first show two weighted reachability problems that are in $\NL$.
\begin{lemma}
\label{shortest-walk-in-nl}
There exists an $\NL$ algorithm which, given a directed graph $G = (V, E)$, vertex weights $l: V \rightarrow \mathbb{Z}^{\leq \poly(n)}$, sets $X_1, X_2 \subseteq V$, and an integer $L \in \mathbb{Z}^{\leq \poly(n)}$, decides whether there exists in $G$ an $X_1 - X_2$ walk of weight at most $L$ and hop-length at most $|V|$.
\end{lemma}
\begin{proof}
    We non-deterministically explore the graph with a walk. First, we non-deterministically pick the first vertex $v_1$ of the walk. If $v_1 \notin X_1$, we reject. Assume that at stage $j$ the algorithm has explored the walk $\{v_1, v_2, \dots, v_j\}$. We store the last vertex $v_j$, the weight of the walk $cw = \sum_{i = 1}^j l(v_i)$, and the hop-length $j-1$. There are now three cases:
    \begin{enumerate}
        \item If $v_j \in X_2$ and $cw \leq L$, we accept.
        \item Else if $j \geq |V|$, we reject. 
        \item Else, we non-deterministically pick a vertex $v_{j+1}$. If $(v_j, v_{j+1}) \notin E$, we reject. Otherwise we continue to the next stage.
    \end{enumerate}

    \noindent
    This procedure always terminates eventually because, if case $1$ is never reached, then $j$ grows by $1$ at each stage and eventually exceeds $|V|$, at which point case $2$ is reached. It accepts if and only if it finds a walk whose weight is at most $L$ and hop-length at most $|V|$.
\end{proof}

\begin{corollary}
\label{shortest-path-in-nl}
There exists an $\NL$ algorithm which, Given a directed graph $G = (V, E)$, vertex weights $l: V \rightarrow \mathbb{Z}^{\leq \poly(n)}$, sets $X_1, X_2 \subseteq V$, such that $G$ does not have any negative weight cycles with respect to $l$, decides whether there exists an $X_1 - X_2$ simple path of weight at most $L$.
\end{corollary}
\begin{proof}
    If all cycles in the graph have non-negative weight, then it is easy to see that there exists a simple path of weight $\leq L$ if and only if there exists a walk of weight $\leq L$ and hop-length $\leq |V|$. Thus, we can simply apply \Cref{shortest-walk-in-nl} to solve this in $\NL$. 
\end{proof}

\noindent
We deduce the following catalytic algorithm.
\begin{lemma}
\label{shortest-path-in-cl}
    There exists a $\CL$ algorithm which, given a directed graph $G = (V, E)$ with vertex weights $l: V \rightarrow \mathbb{Z}^{\leq \poly(n)}$ such that $G$ does not have any negative weight cycles with respect to $l$, along with sets $X_1, X_2 \subseteq V$, computes the minimum weight of a simple $X_1$-$X_2$ path, or concludes that such a path does not exist.
\end{lemma}
\begin{proof}
    Since $\NL \subseteq \CL$~\cite{BuhrmanCleveKouckyLoffSpeelman14}, the $\CL$ algorithm can simply iterate through all possible $L$, starting from $L = 0$, in ascending order, and pick the first one such that \Cref{shortest-path-in-nl} returns $\mathsf{true}$. If \Cref{shortest-path-in-nl} does not return true for any $L \leq \sum_{v \in V} l(v) = \poly(n)$, then our $\CL$ algorithm can conclude that there does not exist any $X_1$-$X_2$ path in $G$. 
 \end{proof}

\noindent
Finally, we present a catalytic subroutine to find the minimum weight and minimum weight hop-length simple cycle in a graph. This algorithm proceeds by a folklore reduction to minimum weight bipartite perfect matching, and then uses the algorithm of~\cite{AgarwalaMertz25} as a black-box. For succinctness, we postpone the proof to the Appendix.
\begin{lemma}
\label{computing-min-weight-cycle-in-cl}
There exists a $\CL$ algorithm which, given a directed graph $G = (V, E)$ with vertex weights $l:V \rightarrow \mathbb{Z}^{\leq \poly(n)}$, such that $G$ does not have any negative weight cycles with respect to $l$, along with a special vertex $c \in V$, computes a minimum weight minimum hop-length simple cycle of $G$ containing $c$.
\end{lemma}

%% file: structure_paper/unique_minimum_weight_matching.tex
\section{Isolation Lemma and Minimum Weight \texorpdfstring{$k$}{k} Linear Matroid Intersection}
\label{isolation_lemma}
Recent advancements in the complexity of linear matroid intersection rely on the isolation lemma (see~\cite{MulmuleyVaziraniVazirani87,NarayananSaranVazirani94, GurjarThierauf17}). This lemma is crucial because it allows these problems to be solved in parallel by demonstrating that a random weight assignment will, with high probability, yield a unique optimal solution.

\begin{definition}[Isolation]
    Let $U$ be a set and $W\colon U\mapsto \mathbb{Z}$. Let $\mathcal{F}$ be a family of subsets of $U$. We say that $W$ isolates a set $S\in \mathcal{F}$ from $\mathcal{F}$ if for all sets $T\in \mathcal{F}, T\neq S$, we have $W(T)>W(S)$.
\end{definition}

\noindent
It was proven in \cite{MulmuleyVaziraniVazirani87}, that given a bipartite graph, an edge $e$, and an edge weight assignment $W$ with the promise that $W$ isolates a perfect matching $M$ of $G$, there is an $\Logspace^{\operatorname{DET}}$ machine (where $\operatorname{DET}$ is the matrix determinant problem) which decides if $e \in M$. A similar statement was proven for linear matroid intersection in \cite{NarayananSaranVazirani94}. \\

\noindent
Observing that $\mathsf{DET} \in \TCo \subseteq \CL$, it was proven in~\cite{AgarwalaMertz25} that given a weight assignment $W$ which isolates a size $k$ matching $M_k$ of $G$, there exists a $\CL$ algorithm which outputs this matching. In this section, we will prove the same statement for linear matroid intersection. \\

\begin{theorem}[\cite{NarayananSaranVazirani94}]
\label{computing-perfect-common-independent-set-is-in-cl}
There exists a $\CL$ algorithm which:
\begin{enumerate}
    \item Takes as input weighted linear matroids $M_1 = (S, \mathcal{I}_1), \ M_2 = (S, \mathcal{I}_2), \ W\colon S \rightarrow \mathbb{Z}^{\leq \poly(n)}$ with linear representations $L_1$ and $L_2$ respectively, both with dimensions $m \times n$.
    \item If the minimum weight perfect common independent set $I$ (i.e. $\lvert I\rvert =m$) of $M_1$ and $M_2$ is unique with respect to $W$, it produces as output the set $I$.
\end{enumerate}
\end{theorem}
\begin{proof}
The only computationally heavy subroutines used in the algorithm of \cite{NarayananSaranVazirani94} are:
\begin{enumerate}
    \item Computing the product of $\poly(n) \times \poly(n)$ dimension matrices with $\poly(n)$ bit entries. This is in $\Logspace \subseteq \CL$. 
    \item Computing the determinant of a $\poly(n) \times \poly(n)$ dimension matrix with $\poly(n)$ bit entries. This is in $\mathsf{GapL} \subseteq \TCo \subseteq \CL$ \cite{MahajanVinay97}.  
    \item Interpolating the $\poly(n)$-bit coefficients of a $\CL$ computable univariate polynomial of degree $\poly(n)$. It is known that this reduces to multiplying the inverse of a $\poly(n) \times \poly(n)$ dimension Vandermonde matrix having $\poly(n)$ bit entries, with a $\poly(n)$ sized vector consisting of $\poly(n)$ bit entries. Computing the inverse of a matrix reduces to computing its determinant and its cofactors, so this is in $\TCo \subseteq \CL$ by combining the previous two facts. 
\end{enumerate}

\noindent
These facts combine to show that the algorithm described in Theorem 4.1 of \cite{NarayananSaranVazirani94} can be implemented in $\CL$. We note a slight technicality: the algorithm in \cite{NarayananSaranVazirani94} requires the weights $W$ to be non-negative. This, however, is easy to handle: we can simply define $W':S \rightarrow Z^{+, \ \leq \poly(n)}$ as $W'(s) = W(s) + |\min_{s'\in S} W(s')| + 1$. The weights $W'$ are strictly positive, and the minimum weight perfect common independent set with respect to $W'$ is unique and equal to $I$.
\end{proof}

\noindent
We now reduce the computation of an isolated size $k$ common independent set to the computation of an isolated perfect common independent set. 

\begin{lemma}
\label{size-k-common-independent-set-reduces-to-perfect}
There exists a $\Logspace$ machine which:
\begin{enumerate}
    \item Takes as input weighted linear matroids $M_1 = (S, \mathcal{I}_1),\ M_2 = (S, \mathcal{I}_2),\ W\colon S \rightarrow Z^{\leq \poly(n)}$ with linear representations $L_1$ and $L_2$.
    \item Computes weighted linear matroids $M_1' = (S \cup S', \mathcal{I}'_1),\ M_2' = (S \cup S', \mathcal{I}'_2),\ W':S \cup S' \rightarrow Z^{\leq \poly(n)}$ such that if $W$ isolates a size $k$ common independent set $I_k$ of $M_1$ and $M_2$, then $W'$ isolates a perfect common independent set $I$ of $M_1'$ and $M_2'$ which satisfies $I \cap S = I_k$.
\end{enumerate}
\end{lemma}

\noindent
The proof of this lemma is technical but not particularly insightful. Thus, we defer its proof to the Appendix. 
\begin{lemma}
\label{computing-size-k-isolated-common-independent-set-is-in-cl}
There exists a $\CL$ algorithm which, given as input linear matroids $M_1 = (S, \mathcal{I}_1), \ M_2 = (S, \mathcal{I}_2), \ W\colon S \rightarrow \mathbb{Z}^{\leq \poly(n)}$, and $k \in [|S|]$ such that $W$ isolates a size $k$ common independent set $I$ of $M_1$ and $M_2$, computes $I$. 
\end{lemma}
\begin{proof}
We can simply apply the reduction in \Cref{size-k-common-independent-set-reduces-to-perfect}, and then the algorithm in \Cref{computing-perfect-common-independent-set-is-in-cl}, in order to obtain $I$.
\end{proof}

\noindent
Finally let us present a slight modification of \Cref{computing-size-k-isolated-common-independent-set-is-in-cl}. Our compression and decompression scheme will need to be able to compute an isolated size $k$ common independent set $I_k$ without having access to the weight of a `compressed' element $s$. The following lemma shows that this is possible in $\CL$.

\begin{lemma}
\label{size-k-included-unique-common-independent-set}
Let $M_1 = (S, \mathcal{I}_1),\ M_2 = (S, \mathcal{I}_2),\ W\colon S \rightarrow \mathbb{Z}^{\leq \poly(n)}$ be weighted linear matroids, and $k \in [|S|]$ be such that $w$ isolates a size $k$ common independent set $I_k$ of $M_1$ and $M_2$.\\

\noindent
Let $s\in S$ be a `compressed' element and let $b=\Ind{s \in I_k}$.\\

\noindent
There exists a $\CL$ algorithm which, given as input the matroids $M_1$ and $M_2$, $k$, the element $s$, the compressed weight assignment $W\restriction_{S \setminus s}$, and the boolean $b$, computes $I_k$.
\end{lemma}
\begin{proof}
Let $|W| = \sum_{x \in S \setminus s} |W(x)|$. Consider the following weight function $W'$:
\[ W'(x) = \begin{cases} 
      W(x) & x \neq s \\
      |W| \cdot (-1)^b & x = s \\
   \end{cases}
\]

\noindent
Clearly, $W'(x) \leq \poly(n)$, and $I_k$ is the unique minimum weight size $k$ common independent set of $M_1$ and $M_2$ under weights $W'$.\\

\noindent
We can now simply apply the algorithm in \Cref{computing-size-k-isolated-common-independent-set-is-in-cl} on $M_1,\ M_2$ with weights $W'$ in order to obtain $I_k$ in $\CL$. 
\end{proof}

%% file: structure_paper/maximal_independent_set.tex
\section{Maximum Common Independent Set, and the Exchange Graph}
\label{maximal_independent_set}
In the previous section, we showed that if a weight assignment $W$ isolates a size $k$ common independent set $I_k$, we can compute $I_k$ in $\CL$.\\

\noindent
Our goal is now to decide whether this set $I_k$ is a maximum sized common independent set. If it is, we can simply output it as our final solution.\\

\noindent
In order to check if $I_k$ is maximum, we will use a standard graph-theoretic tool from matroid theory known as the \textit{exchange graph}. This graph generalizes the use of residual graphs and augmenting paths for maximum bipartite matching. For maximum bipartite matching, maximality testing can be reduced to a $s$-$t$ reachability problem in the residual graph~\cite{Berge57}. A similar statement is true for matroid intersection in the exchange graph.

\begin{definition}[Exchange Graph]
\label{def_exchange_graph}
     Let $M_1= (S,\mathcal{I}_1),\ M_2= (S,\mathcal{I}_2),\ W\colon S \rightarrow\mathbb{Z}$ be weighted matroids, and let $I\in \mathcal{I}_1\cap \mathcal{I}_2$ be a common independent set. \\

    \noindent
     The exchange graph $\mathcal{E}_{M_1,M_2,I} = (S, E)$ is a vertex weighted directed graph. Let $x \notin I$ and $y \in I$. The edge set $E$ is defined as follows:
     \begin{enumerate}
         \item 
        $(y,x)\in E \iff I-\{y\}\cup \{x\}\in \mathcal{I}_1$
         \item 
         $(x,y)\in E \iff I-\{y\}\cup \{x\}\in \mathcal{I}_2$.
     \end{enumerate}

    \noindent
     The vertex weights $l:S \rightarrow \mathbb{Z}$ are defined by:
     \[ l (s) =  \begin{cases} 
       W(s), & \text{if } s \notin I \\
     -W(s), & \text{if }s \in I \\
        \end{cases}
    \]

    \noindent
     Furthermore, $\mathcal{E}_{M_1,M_2,I}$ has two special vertex sets $X_1$ and $X_2$ defined as follows:
     \[X_1=\{x\in S \setminus I\mid I \cup \{x\} \in \mathcal{I}_1\}\]
     \[X_2=\{x\in S \setminus I\mid I \cup \{x\} \in \mathcal{I}_2\}\]
\end{definition}

\noindent
We observe that we can compute the Exchange Graph in $\CL$.

\begin{lemma}
\label{exchange-graph-construction-is-in-cl}
There exists a $\CL$ algorithm which:
\begin{enumerate}
    \item Takes as input weighted matroids $M_1 = (S, \mathcal{I}_1), \ M_2 = (S, \mathcal{I}_2), \ W\colon S \rightarrow \mathbb{Z}^{\leq\poly(n)}$ such that membership testing for $M_1$ and $M_2$ can be done in $\CL$, along with $I \in \mathcal{I}_1 \cap \mathcal{I}_2$, 
    \item Computes the exchange graph $\mathcal{E}_{M_1, M_2, I}$.
\end{enumerate}
\end{lemma}
\begin{proof}
    There are four parts of $\mathcal{E}$ that the machine needs to compute: the vertices, the edges, the vertex weights, and the special sets $X_1, X_2$:
    \begin{enumerate}
        \item The vertices of $\mathcal{E}_{M_1, M_2, I}$ are simply $S$.
        \item The vertex weight of any vertex $x \in S$ is simply $-W(x)$ if $x \in I$, and $W(x)$ otherwise.
        \item The edges of $G$ can be computed as follows: for $x \in I$ and $y \notin I$, the edge $(x, y)$ is in the graph if and only if $I - x + y \in \mathcal{I}_1$. Since membership testing for $M_1$ is assumed to be in $\CL$, this can be decided in $\CL$. The inclusion of edge $(y, x)$ can be determined in the same way for $M_2$.
        \item $X_1$ and $X_2$ can be computed as follows: $x \in X_i \iff I \cup \{x\} \in \mathcal{I}_i$. Again, this is simply membership testing for $M_1$ and $M_2$.
    \end{enumerate}

\noindent
This completes the proof. 
\end{proof}

\noindent
The exchange graph is a well-known and extensively studied object. Notably, the problem of determining whether a common independent set $I$ is of maximum size can be solved by deciding $s$-$t$ reachability in the exchange graph.
\begin{lemma}[\cite{Edmonds2003,Edmonds2009}]
\label{no_s_t_path_equals_maximal}
    A common independent set $I \in \mathcal{I}_1 \cap \mathcal{I}_2$ is of maximum size if and only if there is no path from $X_1$ to $X_2$ in the exchange graph $\mathcal{E}_{M_1,M_2,I}$.
\end{lemma}

\noindent
An immediate corollary of is the following:
\begin{lemma}
\label{test-if-k+1-exists}
There exists a $\CL$ algorithm which:
\begin{enumerate}
    \item Takes as input weighted matroids $M_1 = (S, \mathcal{I}_1), \ M_2 =  (S, \mathcal{I}_2), \ W\colon S \rightarrow \mathbb{Z}^{\leq \poly(n)}$ such that membership testing for $M_1$ and $M_2$ can be done in $\CL$, along with a size $k$ common independent set $I_k$ of $M_1$ and $M_2$.
    \item Decides whether a size $k+1$ common independent set of $M_1$ and $M_2$ exists.
\end{enumerate}
\end{lemma}
\begin{proof}
The following $\CL$ algorithm works:
\begin{enumerate}
    \item Compute the exchange graph $\mathcal{E}_{M_1, M_2, I_k}$ (\Cref{exchange-graph-construction-is-in-cl}). 
    \item Test whether an $X_1-X_2$ path in $\mathcal{E}_{M_1, M_2, I_k}$ exists. (\Cref{shortest-path-in-cl}). 
\end{enumerate}
\noindent
Correctness follows directly from \Cref{no_s_t_path_equals_maximal}.
\end{proof}

%% file: structure_paper/from_k_to_k_p_1.tex
\section{Checking if \texorpdfstring{$W$}{w} Isolates a Size \texorpdfstring{$k+1$}{k + 1} Common Independent Set}
\label{section_intersecting_set_from_k_to_k_+_1}

\noindent
In the last two sections, we provided catalytic algorithms for computing an isolated size $k$ common independent set $I_k$, and then testing whether $I_k$ is maximum. If not, we know that a size $k+1$ common independent set exists. Our goal in this section is to either certify that $W$ also isolates a size $k+1$ common independent set, or to show how to compress (and later decompress) the catalytic tape if it doesn't.

\subsection{More Facts About the Exchange Graph}
\noindent
We presented in \Cref{no_s_t_path_equals_maximal} a connection between the paths in the exchange graph and the existence of common independent sets of size $k+1$. We now present slightly deeper properties about this connection:

\begin{theorem}[\cite{Brualdi1969,Krogdahl1974,Krogdahl1976,Krogdahl1977, Fujushige1977,Edmonds2003,Edmonds2009}]
\label{main-ext-theorem}

Let $M_1 = (S, \mathcal{I}_1),\ M_2 = (S, \mathcal{I}_2),\ W:S \rightarrow \mathbb{Z}$ be weighted matroids. Let $I$ be a minimum weight, not necessarily unique, size $k$ common independent set of $M_1$ and $M_2$.\\

\noindent
Let $P$ be a minimum weight minimum hop-length path from $X_1$ to $X_2$ in $\mathcal{E}_{M_1, M_2, I_k}$.\\

\noindent
The following is known:

\begin{enumerate}
    \item $\mathcal{E}_{M_1, M_2, I_k}$ does not contain any negative weight cycles.
    \item  $I_{k+1} = I_k \Delta P$ is a minimum weight size $k+1$ common independent set of $M_1$ and $M_2$.
\end{enumerate}
\end{theorem}

\noindent
\Cref{main-ext-theorem} is described in detail in Section 41.3 of the textbook 'Combinatorial Optimization' by Alexander Schrijver \cite{Schrijver03Textbook}. \footnote{Schrijver's description is in terms of maximum weight common independent sets and uses the weight function $-l$. However, it is simple to see that these descriptions are equivalent, simply by substituting the matroid weights $-w$ into Schrijver's theorems.} \\

\noindent
Additionally, in the case where $W$ isolates a size $k$ common independent set, we observe that all cycles must have strictly positive weight.
\begin{lemma}
\label{no_negative_cycles}
Let $M_1 = (S, \mathcal{I}_1),\ M_2 = (S, \mathcal{I}_2),\ W:S \rightarrow \mathbb{Z}$ be weighted matroids such that $W$ isolates a size $k$ common independent set $I_k$ of $M_1$ and $M_2$. $\mathcal{E}_{M_1, M_2, I_k}$ does not contain any weight $0$ cycles.
\end{lemma}
\begin{proof}
Lemma 41.5$\alpha$ from \cite{Schrijver03Textbook} states that if $\mathcal{E}_{M_1, M_2, I_k}$ contains a cycle $C$, then there exists a size $k$ common independent set $I'_k \neq I_k$ such that either $W(I'_k) < W(I_k)$ or $W(I'_k) \leq W(I_k) + l(C)$. Thus, if there were a cycle $C$ such that $l(C) = 0$, it would imply that $I_k$ is not the unique minimum weight size $k$ common independent set. 
\end{proof}

\noindent
Now, we showed in \Cref{isolation_lemma} that, when $W$ isolates a size $k$ common independent set $I_k$, we can construct $I_k$. We need to additionally compute, for any $s \in S$, the (not necessarily unique) minimum weight size $k$ common independent $I^{+s}$ such that $s \in I^{+s}$, and similarly the minimum weight size $k$ common independent $I^{-s}$ such that $s \notin I^{-s}$. Note that either $I^{+s}$ or $I^{-s}$ is $I_k$. In order to compute the other, we need the following fact: 

\begin{lemma}
\label{cycle-shifting-lemma}
Let $M_1 = (S, \mathcal{I}_1)$, $M_2 = (S, \mathcal{I}_2)$, $W:S \rightarrow \mathbb{Z}^{\leq \poly(n)}$ be weighted matroids which admit a unique minimum weight size $k$ common independent set $I_k$. Let $s \in S$.\\

\noindent
Let $C^*$ be the minimum weight minimum hop-length cycle $C$ in $\mathcal{E}_{M_1, M_2, I_k}$ such that $s \in C$. Then $I = I_k \Delta C^*$ is a minimum weight size $k$ common independent set of $M_1$ and $M_2$ such that $s \in I \iff s \notin I_k$. That is, \[I_k \Delta C^* \in \argmin_{I \in \{I \in \mathcal{I}_1 \cap \mathcal{I_2} \ | \ |I| = k, \ s \in I \Delta I_k\}} W(I)\]
\end{lemma}
\begin{proof}
First, we will show that $I_k \Delta C^*$ is a common independent set of $M_1$ and $M_2$. Lemma 41.5$\alpha$ from \cite{Schrijver03Textbook} states that if $I_k \Delta C^* \notin \mathcal{I}_1 \cap \mathcal{I}_2$, then there exists either a negative weight cycle in $\mathcal{E}_{M_1, M_2, I_k}$ (which is impossible by \Cref{main-ext-theorem}), or there exists another cycle $C \subsetneq C^*$ with $l(C) \leq l(C^*)$. The latter would contradict the fact that $C^*$ has minimum hop-length amongst the minimum weight cycles. Therefore, neither of these cases are possible, and hence $I_k \Delta C^*$ is necessarily a common independent set of $M_1$ and $M_2$. Moreover, $(I_k \Delta C^*) \Delta I_k = C^*$. Thus, $s \in (I_k \Delta C^*) \Delta I_k$.\\

\noindent
Now, we will prove the minimality of $I_k \Delta C^*$. Let $I' \in \mathcal{I}_1 \cap \mathcal{I}_2$ be a size $k$ common independent set such that $s \in I' \Delta I_k$. Theorem 41.5 of \cite{Schrijver03Textbook} shows that $I' \Delta I_k$ is the union of disjoint cycles in $\mathcal{E}_{M_1, M_2, I_k}$. Let these cycles be $C_1, \dots, C_m$. Since $s \in I' \Delta I_k$, there must exist a cycle containing $s$, assume that it is $C_1$. $W(I') = W(I_k) + \sum_{i=1}^m l(C_i) \geq W(I_k) + l(C_1) \geq W(I_k) + l(C^*)$. The second to last step was due to the fact that all cycles have non-negative weight in $\mathcal{E}_{M_1, M_2, I_k}$. The last step was due to the fact that $l(C^*)$ is the minimum weight amongst all cycles in $\mathcal{E}_{M_1, M_2, I_k}$ which contain $s$.\\

\noindent
This completes the proof.
\end{proof}

\subsection{Threshold Elements}
\label{threshold_element}
In order to handle the case where $W$ does not isolate a size $k+1$ common independent set, we need to introduce the concept of \textit{threshold elements}.

\begin{definition}[Threshold Elements]
Let $M_1 = (S, \mathcal{I}_1)$, $M_2 = (S, \mathcal{I}_2)$, $W:S \rightarrow \mathbb{Z}^{\leq \poly(n)}$ be weighted matroids which admit a unique minimum weight size $k$ common independent set $I_k$.\\

\noindent
An element $s \in S$ is a $k+1$-threshold element if there exist two size $k+1$ minimum weight common independent sets, $I$ and $I'$, such that $s \in I$ and $s \notin I'$.\\

\noindent
The set of $k+1$-threshold elements of $M_1$ and $M_2$ is denoted by $T^{k+1}$.
\end{definition}

\begin{lemma}[Threshold elements exist $\iff$ the minimum weight size $k+1$ common independent set is not unique]\ \\
\label{threshold-existence}
\noindent
Let $M_1 = (S, \mathcal{I}_1),\ M_2 = (S, \mathcal{I}_2),\ W:S \rightarrow \mathbb{Z}$ be weighted matroids, and let $k \in [|S|]$ such that $W$ isolates a size $k$ common independent set $I_k$ of $M_1$ and $M_2$, and a size $k+1$ common independent set of $M_1$ and $M_2$ exists. The following statements are equivalent:
\begin{enumerate}
    \item The minimum weight size $k+1$ common independent set of $M_1$ and $M_2$ is not unique. 
    \item $|T^{k+1}| \geq 1$.
\end{enumerate}
\end{lemma}
\begin{proof}
This  trivially holds:
\begin{enumerate}
    \item $(1) \implies (2)$: Let $I_1$ and $I_2$ be minimum weight size $k+1$ common independent sets of $M_1$ and $M_2$. Any element $s \in I_1 \setminus I_2$ is a threshold element. Thus, $|\mathcal{I}^{k+1}_{\min}| > 1 \implies |T^{k+1}_{M_1, M_2}| \geq 1$.

    \item $(2) \implies (1)$: By definition of $T^{k+1}_{M_1, M_2}$, for any $s \in T^{k+1}_{M_1, M_2}$, there exist $X, X' \in \mathcal{I}^{k+1}_{\min}$ such that $X \neq X'$. Thus, $|T^{k+1}_{M_1, M_2}| \geq 1 \implies |\mathcal{I}^{k+1}_{\min}| > 1$ 
\end{enumerate}

\noindent
This completes the proof.
\end{proof}

\noindent
 The goal is therefore simply to decide whether a threshold element exists, and if it does then to find one. In order to do this, first observe that we can define a threshold element in terms of inclusion and exclusion matroids.
\begin{observation}
\label{observation_threshold_element}
 $s \in T^{k+1} \iff \min_{k}(M_1 \restriction_s, M_2 \restriction_s) + w(s) = \min_{k+1}(M_1 - \{s\}, M_2 - \{s\})$.
\end{observation}

\noindent
We will show that for any $s \in S$, both sides of this equation can be computed in $\CL$. Thus, we can simply iterate over all $s \in S$, compute both sides of this equation, and test if equality holds - this suffices to both decide whether a threshold element exists, and to find one if it does.

\subsection{The Hunt for Threshold Elements}
\label{finding_threshold_element}
In this subsection, we will describe how to decide in $\CL$ if an element $s \in S$ is a threshold element, and thus how to find a threshold element if it exists. By the equation in \Cref{observation_threshold_element}, this reduces to computing for a fixed $s$ the weights $a = \min_{k}(M_1 \restriction_s, M_2 \restriction_s)$, $b = \min_{k+1}(M_1 - \{s\}, M_2 - \{s\})$, and then simply checking if $b=a+w(s)$.\\

\noindent
For this, let us first make the following observation:

\begin{lemma}
\label{compute-min-k+1}
There exists a $\CL$ algorithm which, given as input weighted matroids $M_1 = (S, \mathcal{I}_1), \ M_2 =  (S, \mathcal{I}_2), \ W\colon S \rightarrow \mathbb{Z}^{\leq \poly(n)}$ such that membership testing for $M_1$ and $M_2$ can be done in $\CL$, along with a minimum weight size $k$ common independent set $I_k$, computes $\min_{k+1}(M_1, M_2)$.
\end{lemma}
\begin{proof} The following algorithm works:
\begin{enumerate}
    \item Compute $\mathcal{E}_{M_1, M_2, I_k}$ (\Cref{exchange-graph-construction-is-in-cl}).
    \item Compute the minimum weight of an $X_1$-$X_2$ path in $\mathcal{E}_{M_1, M_2, I_k}$ (\Cref{shortest-path-in-cl}). Let it be $L$.
    \item Output $w(I_k) + L$.
\end{enumerate}

\noindent
Correctness follows from \Cref{main-ext-theorem}.
\end{proof}

\noindent
Now, note that $\min_{k+1}(M_1 - \{s\}, M_2 - \{s\})$ is simply the minimum weight of a size $k+1$ common independent set of $M_1$ and $M_2$ \textit{not containing} $s$. \Cref{compute-min-k+1} tells us that in order to compute this, it suffices to compute the minimum weight size $k$ common independent set not containing $s$ - let such a set be $I^{-s}$. Similarly, $\min_{k}(M_1 \restriction_s, M_2 \restriction_s)$ is simply the minimum weight of a size $k+1$ common independent set \textit{containing} $s$, but excluding the weight of $s$ itself. \Cref{compute-min-k+1} tells us that in order to compute this, it suffices to compute the minimum weight size $k$ common independent set containing $s$ - let such a set be $I^{+s}$. Simply note that either $I^{+s}$ or $I^{-s}$ is $I_k$. \Cref{cycle-shifting-lemma} implies that the other can be computed by finding a minimum weight cycle $C$ in $\mathcal{E}_{M_1, M_2, I_k}$ using \Cref{computing-min-weight-cycle-in-cl}, and then taking the symmetric difference of $I_k$ and $C$.

\begin{lemma}
\label{size-k-restriction-compression}
Let $M_1 = (S, \mathcal{I}_1)$, $M_2 = (S, \mathcal{I}_2)$, $W:S \rightarrow \mathbb{Z}^{\leq \poly(n)}$ be weighted linear matroids which admit a unique minimum weight size $k$ common independent set $I_k$. Let $s \in S$ and let $b = \Ind{s \in I_k}$ indicate whether $s \in I_k$.\\

\noindent
Given $M_1, \ M_2, \ W\restriction_{S \setminus s},\ k, \ s, \ b $, there exists a $\CL$ algorithm which either:
\begin{enumerate}
    \item Outputs a minimum weight size $k$ common independent set $I'_k$ of $M_1$ and $M_2$ such that $s \in I'_k \iff s \notin I_k$, or
    \item Certifies that such a common independent set does not exist.
\end{enumerate}
\end{lemma}
\begin{proof}

\noindent
Consider $W':S \rightarrow \mathbb{Z}^{\leq \poly(n)}$ such that $W'(x) = W(x)$, for $x \neq s$, and $W'(s) = (1 + \sum_{x \in S, x \neq s} |W(s)|) \cdot (-1)^{b}$. \\

\noindent
The minimum weight size $k$ common independent set of $M_1$ and $M_2$ with respect to $W'$ is unique and must be $I_k$. Thus, it can be computed using \Cref{computing-size-k-isolated-common-independent-set-is-in-cl}.\\

\noindent
For every $S_1, S_2 \subseteq S$ such that $s \in S_1 \iff s \in S_2$, we have $W'(S_1) \leq W'(S_2) \iff W(S_1) \leq W(S_2) $. Thus, the weight size $k$ common independent set $I$ of $M_1$ and $M_2$ under $W'$ such that $s \in I \iff s \notin I_k$ is the same as that under $W$, and must be $I'_k$ (though $I'_k$ need not be unique). The algorithm is fairly straightforward:

\begin{enumerate}
    \item Compute $I_k$ (\Cref{computing-size-k-isolated-common-independent-set-is-in-cl}).
    \item Compute $\mathcal{E}_{M_1, M_2, I_k}$ under weights $W'$ (using \Cref{exchange-graph-construction-is-in-cl}). Let the lengths be $l'$.
    \item Compute a minimum weight minimum hop-length cycle $C^*$ of $\mathcal{E}_{M_1, M_2, I_k}$ containing $s$ under weights $l'$ (using \Cref{computing-min-weight-cycle-in-cl}). If no such cycle exists, certify the second case of the statement.
    \item Else, output $C^* \Delta I_k$. 
\end{enumerate}

\noindent
The correctness of the algorithm follows from \Cref{cycle-shifting-lemma}.
\end{proof}

\begin{lemma}
\label{size-k-restriction-and-inclusion-compression-1}
Let $M_1 = (S, \mathcal{I}_1)$, $M_2 = (S, \mathcal{I}_2)$, $W:S \rightarrow \mathbb{Z}^{\leq \poly(n)}$ be weighted linear matroids which admit a unique minimum weight size $k$ common independent set $I_k$. Let $s \in S$ and $b = \Ind{s \in I_k}$.\\
\noindent
Given $M_1, \ M_2, \ W\restriction_{S \setminus s}, \ s, \ b = \Ind{s \in I_k}$, there exists a $\CL$ algorithm which either:
\begin{enumerate}
    \item Outputs $\min_{k}(M_1\restriction_s, M_2 \restriction_s)$, or
    \item Certifies that a size $k$ common independent set of $M_1\restriction_s$ and $M_2 \restriction_s$ does not exist.
\end{enumerate}

\noindent
Similarly, there exists another $\CL$ algorithm which either:
\begin{enumerate}
    \item Outputs $\min_{k+1}(M_1 - \{s\}, M_2 - \{s\})$, or
    \item Certifies that a size $k+1$ common independent set of $M_1 - \{s\}$ and $M_2 - \{s\}$ does not exist.
\end{enumerate}
\end{lemma}
\begin{proof}
\noindent
First for the inclusion matroid. Note that any $J \subseteq S \setminus \{s\}$ is a minimum weight size $k-1$ common independent set of $M_1 \restriction_s$ and $M_2 \restriction_s$ if and only if $I = J \cup \{s\}$ is a minimum weight size $k$ common independent set of $M_1$ and $M_2$ including $s$.\\

\noindent
Thus, the problem of computing a minimum weight size $k-1$ common independent set of $M_1 \restriction_s$ and $M_2 \restriction_s$ reduces simply to computing $I$ (or certifying that it does not exist). This can be done as follows:
\begin{enumerate}
    \item Compute $I_k$ (\Cref{computing-size-k-isolated-common-independent-set-is-in-cl}).
    \item If $b = 1$, we have $I = I_k$.
    \item If $b = 0$, we can compute $I$ using \Cref{size-k-restriction-compression}. If  no such set exists, we can certify the second case of the statement.
\end{enumerate}

\noindent
Now, we have $J = I \setminus \{s\}$, a minimum weight size $k-1$ common independent set of $M_1 \restriction_s$ and $M_2 \restriction_s$. We can now:
\begin{enumerate}
    \item Test if a size $k$ common independent set of $M_1 \restriction_s$ and $M_2 \restriction_s$ exists (using \Cref{test-if-k+1-exists}). If not, we can certify the second case of the statement.
    \item Else, we can compute and output $\min_{k}(M_1\restriction_s, M_2 \restriction_s)$ (using \Cref{compute-min-k+1}).
\end{enumerate}

\noindent
For the exclusion matroid case, the algorithm is similar:
\begin{enumerate}
    \item Compute $I_k$ (using \Cref{size-k-included-unique-common-independent-set}). 
    \item Compute a minimum weight size $k$ common independent set $I'_k$ of $M_1 - \{s\}$ and $M_2 - \{s\}$:
    \begin{enumerate}
        \item If $b = 0$, this is exactly $I_k$.
        \item If $b = 1$, this is a minimum weight size $k$ common independent set of $M_1$ and $M_2$ which does not include $s$. This can be computed using \Cref{size-k-restriction-compression}. If such a set does not exist, we can certify the second case of the statement.
    \end{enumerate}
    \item Test if a size $k+1$ common independent set of $M_1 - \{s\}$ and $M_2 - \{s\}$ exists (using \Cref{test-if-k+1-exists}). If such a set does not exist, we can certify the second case of the statement.
    \item Compute $\min_{k+1}(M_1 - \{s\}, M_2 - \{s\})$ (using \Cref{compute-min-k+1}).
\end{enumerate}

\noindent
This completes the proof.
\end{proof}

\noindent
\Cref{size-k-restriction-and-inclusion-compression-1} enables us to compute both side of the equation in \Cref{observation_threshold_element}. We now need to check, for all $s\in S$, if $s$ satisfies this equation. If it does, $s$ will be our threshold element.

\begin{lemma}
\label{compress-or-certify}
There exists a $\CL$ algorithm which, given weighted linear matroids $M_1 = (S, \mathcal{I}_1), \ M_2 =  (S, \mathcal{I}_2), \ W\colon S \rightarrow \mathbb{Z}^{\leq \poly(n)}$ such that the minimum weight size $k$ common independent set of $M_1$ and $M_2$, $I_k$, is unique, and a size $k+1$ common independent set of $M_1$ and $M_2$ exists:
\begin{enumerate}
    \item Computes $s \in T^{k+1}$. Or
    \item Certifies that the minimum weight size $k+1$ common independent set of $M_1$ and $M_2$ is unique.
    
\end{enumerate}
\end{lemma}
\begin{proof}

For all $s\in S$, we will do the following test:
\begin{enumerate}
    \item Compute $I_k$ (Using \Cref{computing-size-k-isolated-common-independent-set-is-in-cl}).
    \item Compute $\mathcal{E}_{M_1, M_2, I_k}$ (Using \Cref{exchange-graph-construction-is-in-cl}).
    \item Compute $a = \min_{k}(M_1 \restriction_s, M_2 \restriction_s)$ and $b = \min_{k+1}(M_1 - \{s\}, M_2 - \{s\})$ (Using \Cref{size-k-restriction-and-inclusion-compression-1}). 
    \item If either $a$ or $b$ doesn't exist, continue to the next $s$.
    \item Else, check if $b = a + W(s)$. If yes, then $s \in T^{k+1}$, so we can output $s$ and terminate. Else, $s \notin T^{k+1}$, so we can continue to the next $s$.
\end{enumerate}

\noindent 
The existence of a threshold element $s$ is guaranteed by \Cref{threshold-existence} if and only if the minimum weight size $k+1$ common independent set of $M_1$ and $M_2$ is not unique. Thus, if such an $s$ is found, we can simply output it. If no such $s$ is found, we can certify that the minimum weight size $k+1$ common independent set of $M_1$ and $M_2$ is unique. 
\end{proof}

%% file: structure_paper/final_algorithm.tex
\section{Final Algorithm}
\noindent
We are now ready to describe our compression and decompression algorithms. 
\label{final_algorithm}
\begin{lemma}
\label{small-compression}
Let $M_1 = (S, \mathcal{I}_1)$ and $M_2 = (S, \mathcal{I}_2)$ be linear matroids given as input.\\
\noindent 
Let $(w, r, \tau)$ be a catalytic tape where:
\begin{enumerate}
    \item 
    $\lvert w\rvert= \lvert S\rvert \cdot 10 \log \lvert S\rvert $. This section of the catalytic tape is interpreted as a weight assignment $w:S \rightarrow \mathbb{Z}^{\leq |S|^{10}}$.
    \item $\lvert r\rvert = 10 \log \lvert S\rvert$. This section of the catalytic tape is interpreted as a `reserve' weight in $\mathbb{Z}^{\leq |S|^{10}}$. 
    \item $|\tau| = \poly(n)$. This section of the catalytic tape has no special interpretation, it simply catalytic space to be used by our catalytic subroutines.
\end{enumerate}  

\noindent 
There exist a pair of $\CL$ algorithms $\mathcal{C}omp$ and $\mathcal{D}ecomp$ with the following behaviour:
\begin{enumerate}
    \item
    $\mathcal{C}omp$, when run on inputs $M_1$, $M_2$, $k \in [|S|]$, $s \in S$ with the catalytic tape containing $(w, r, \tau)$, such that $w$ does not isolate a size $k+1$ common independent set of $M_1$ and $M_2$ ($k$ is the minimum such size), and $s \in T^{k+1}$, outputs nothing but changes the catalytic tape to a string $(w', r', \tau)$ such that
    \begin{enumerate}
        \item 
        $w'(e)=w(e)$ for all $e\neq s$, and $w'(s)=r$.
        \item 
        $r'=(0^{8\log|S| - 1}, s, k, b = \Ind{s \in I_k})$
    \end{enumerate}

    \item $\mathcal{D}ecomp$, when run on inputs $M_1$ and $M_2$ with the catalytic tape $(w', r', \tau)$ outputs nothing, but returns the catalytic tape to its original state $\tau$. 
\end{enumerate}
\end{lemma}
\begin{proof}
    $\mathcal{C}omp$ simply swaps $w(s)$ with $r$, and then writes $(0^{8\log|S| - 1}, s, k, b = \Ind{s \in I_k})$ in the place of $r$.\\
    
\noindent
We now focus on the $\mathcal{D}ecomp$ procedure. Let us first make the following observation. 
    \begin{lemma}
\label{decompression-function}
Let $M_1 = (S, \mathcal{I}_1),\ M_2 = (S, \mathcal{I}_2),\ W:S \rightarrow \mathbb{Z}^{\leq \poly(n)}$ be weighted linear matroids such that the minimum weight size $k$ common independent set of $M_1$ and $M_2$, $I_k$, is unique, and a size $k+1$ common independent set of $M_1$ and $M_2$ exists.\\

\noindent
Given $M_1, \ M_2,\ s \in T^{k+1}, \ b = \Ind{s \in I_k}$, and $W \restriction_{S \setminus \{s\}}$: $\min_{k+1}(M_1, M_2) - w(s)$ and $\min_{k+1}(M_1, M_2)$ can be computed in $\CL$.
\end{lemma}
\begin{proof}
The following works:
\begin{enumerate}
    \item For $\min_{k+1}(M_1, M_2) - w(s)$, simply note that this value is exactly $\min_{k}(M_1 \restriction_s, M_2 \restriction_s)$, which can be computed using in \Cref{size-k-restriction-and-inclusion-compression-1}.
    \item On the other hand, $\min_{k+1}(M_1, M_2)$ is exactly $\min_{k+1}(M_1 - \{s\}, M_2 - \{s\})$, which can be computed again using \Cref{size-k-restriction-and-inclusion-compression-1}.
\end{enumerate}
\end{proof}

\noindent
Recall that $r'=(0^{8\log|S| - 1}, s, k, b = \Ind{s \in I_k})$. Therefore, using \Cref{decompression-function}, we can compute $a=\min_{k+1}(M_1,M_2)-w(s)$ and $b=\min_{k+1}(M_1,M_2)$ from the information stored in $r'$, and thus compute $w(s)=b-a$. Once we have $w(s)$ on the work tape, we can swap $r'$ and $w'(s)$, and then swap $w'(s)$ and $w(s)$. This is guaranteed to return the catalytic tape to its original state.\\

\noindent
This completes the proof.
\end{proof}

\noindent
We are now ready to present our final algorithm. The idea is to proceed in stages: in each stage, if the weight assignment $w$ on the catalytic tape isolates common independent sets of all sizes,  we can compute and output a maximum sized common independent set. Else there exists $k$ such that a size $k+1$ common independent set exists, but is not isolated by $w$. In this case we can use our compression procedure on the catalytic tape. After $\poly(n)$ stages, we will have either solved the problem, or freed up polynomial space on the catalytic tape, at which point we can run Edmond's polynomial time algorithm~\cite{Edmonds2003,Edmonds2009, Edmonds1979} directly on the catalytic tape.\\

\noindent
Let us present this $\CL$ algorithm in detail.

\begin{theorem}[Restatement of \Cref{thm:lmi-in-cl}]
\label{final-algorithm}
$\lmi$ is in $\CL$.
\end{theorem}
\begin{proof}
Let $M_1 = (S, \mathcal{I}_1)$ and $M_2 = (S, \mathcal{I}_2)$ be linear matroids given as input, and let us define $n=\lvert S\rvert $.\\

\noindent
Let $T$ be the time taken by a strongly polynomial time algorithm $\mathcal{A}$ to solve linear matroid intersection on $M_1$ and $M_2$ (we only need $T$ to be a $\poly(n)$ upper bound, so it is easy to compute in $\CL$).\\

\noindent
The catalytic tape is partitioned as follows:
    \begin{itemize}
        \item 
        The first section comprises of $10 \cdot n\log{n}$ bits. It is interpreted as $n$ blocks of $10\log{n}$ bits each, which correspond to a weight assignment $W:S \rightarrow \mathbb{Z}^{\leq n^{10}}$.
        \item 
        The second section will hold $T \cdot 10 \log n$ bits. It is interpreted as $T$ blocks of $10 \log n$ bits each, each block corresponding to an element of $\mathbb{Z}^{\leq n^{10}}$. Let these be $r_1, \dots, r_T$
        \item 
        The final section $\tau$ has no special interpretation. It is simply catalytic space to be used by our catalytic subroutines.
    \end{itemize}

\noindent
On the work tape we maintain two counters $k\in [n]$ and $j\in [T]$. $k$ is initialised to $0$ and $j$ is initialised to $1$.\\

\noindent
The algorithm iterates through $k$ while maintaining the invariant that the weights $W$ isolate a size $k$ common independent $I_k$ of $M_1$ and $M_2$. At each step of the iteration, we do the following:
\begin{enumerate}
    \item 
    Check maximality.  
    
    \noindent
    The algorithm first checks if $I_k$ is a maximum sized common independent set of $M_1$ and $M_2$ by computing it using \Cref{computing-size-k-isolated-common-independent-set-is-in-cl} and then using the algorithm in \Cref{test-if-k+1-exists}. If it is of maximum size, the algorithm simply outputs $I_k$.
    
\item 
Check if $W$ isolates a $k+1$-size common independent set. 

\noindent
If $I_k$ is not a maximum sized common independent set of $M_1$ and $M_2$, the algorithm checks if $W$ isolates a size $k+1$ common independent set, or else finds a threshold element $s \in T^{k+1}$ using \Cref{compress-or-certify}. If the algorithm does not find any threshold element, it can simply increment $k$ and return to step 1. Else, it proceeds to the next step.
\item 
Compression.

\noindent
We have a threshold element $s \in T^{k+1}$. We can now use the $\mathcal{C}omp$ procedure defined in \Cref{small-compression}, with the catalytic tape $(W, r_j, \tau)$. The weight of $s$ on the catalytic tape is replaced with $r_j$, and more importantly the place of $r_j$ now has the first $> 7\log n$ bits equal to $0$. We increment the counter $j$ by 1, set $k$ to $0$, and return to step $1$.
\end{enumerate}

\noindent
Each iteration either increments $k$ or increments $j$, the latter of which increases monotonically. If $k$ reaches $|S|$ at any point, we have a maximum sized common independent set of $M_1$ and $M_2$ by definition, and step 1 will output this. Else, if $j$ exceeds $T$, we have compressed each $r_i$ for $i \in [T]$. This means that the first $7\log n$ bits of each block $r_i$ are all $0$s. This acts as free space, which we can use to run the algorithm $\mathcal{A}$ and solve linear matroid intersection, and output the result. Once this is done, we can proceed to the next part.\\

\noindent
Decompression.

\noindent 
The algorithm simply needs to revert the changes it made to the catalytic tape via the $\mathcal{C}omp$ procedure. It does this by simply iteratively decrementing $j$, and then running the $\mathcal{D}ecomp$ procedure on $M_1$ and $M_2$ using the catalytic tape $(w, r_j, \tau)$, until $j$ reaches 0, at which point the catalytic tape is fully restored.
\end{proof}

%% file: structure_paper/conclusion.tex
\section{Conclusion and Open Problems}

In this paper, we solve $\lmi$ in catalytic logspace by derandomizing the isolation lemma, extending the bipartite matching result of~\cite{AgarwalaMertz25}. We thus present the hardest problem yet known to be in $\CL$. A natural question is if we can solve harder problems in $\CL$. We briefly discuss two candidates:

\begin{enumerate}
    \item Linear matroid parity. This problem admits a deterministic polynomial time algorithm~\cite{Lovasz1978,Lovasz1979,Lovasz1980}, thus a $\CL$ algorithm may be the next logical step. Just as bipartite matching is a special case of linear matroid intersection, non-bipartite matching is a special case of the linear matroid parity problem. 
    \item Exact linear matroid intersection. This problem admits a randomized polynomial time algorithm~\cite{MulmuleyVaziraniVazirani87}, but is not known to be in $\P$. Thus, a $\CL$ algorithm for this would show a strong barrier towards proving $\CL \subseteq \P$. Exact matching is a special case of this problem. 
\end{enumerate}

\noindent
Additionally, it would be interesting if resources other than catalytic space, such as nondeterminism, could lead to similar sublinear free space and polynomial time algorithms for linear matroid intersection, or even bipartite matching. 

%% file: structure_paper/appendix.tex
\section{Appendix}
\begin{lemma*}[\Cref{computing-min-weight-cycle-in-cl}]

There exists a $\CL$ algorithm which, given a directed graph $G = (V, E)$ with vertex weights $l:V \rightarrow \mathbb{Z}^{\leq \poly(n)}$, such that $G$ does not have any negative weight cycles with respect to $l$, along with a special vertex $c \in V$, computes a minimum weight minimum hop-length simple cycle of $G$ containing $c$.
\end{lemma*}

\begin{proof}

\noindent
First, we transform the vertex weights $l$ to edge weights $W:E \rightarrow \mathbb{Z}^{\leq \poly(n)}$. Any edge $e$ entering vertex $v$ gets weight $W(e) = l(v)$. Now note that for any cycle $C$, we have $W(C) = l(C)$, because for every vertex of the cycle, there is exactly one edge in the cycle entering it. Thus the problem is equivalent with edge weights $W$ instead of vertex weights $l$. \\

\noindent
We will construct the graph $G' = (V', E')$ with edge weights $w'$. First we define \[V' = \{(v, i) \ | \ v \in V, \ i \in \{0, 1\}\}\] 

\noindent
Now, we define 
\[E_1 = \{((v_1, 0), (v_2, 1)) \ | \ (v_1, v_2) \in E\}\]
\[E_2 = \{((v, 0), (v, 1)) \ | \ v \in V \setminus \{c\}\}\]
\[E' = E_1 \cup E_2\]

\noindent
Finally, for $e = ((v_1, 0), (v_2, 1)) \in E_1$, we define
\[w'(e) = W((v_1, v_2)) \cdot (n+1) + 1\]

\noindent
For $e = ((v, 0), (v, 1)) \in E_2$, we define
\[w'(e) = 0\]

\noindent
Let $M$ be the minimum weight perfect matching on the bipartite graph $G' = (V', E')$ with weights $w'$. For an edge in $e \in E_1$, let $o(e)$ refer to the corresponding edge in $E$. We claim that $o(M \cap E_1) = \bigcup_{e \in M \cap E_1} o(e)$ is a minimum weight minimum hop-length cycle of $G$ containing $c$.\\

\noindent
First, note that for any cycle $C$ containing $s$, $\bigcup_{e \in C}o^{-1}(e) \  \cup \ \{((v, 0), (v, 1)) \ | \ v \notin C\}$ is a perfect matching of $G'$ of weight $W(C) \cdot (n+1) + |C|$. Let $C^*$ be a minimum weight minimum hop-length cycle of $G$. This implies that the minimum weight perfect matching of $G$ has weight at most $W(C^*) \cdot (n+1) + |C^*|$. \\

\noindent
Furthermore, for any perfect matching $M$ of $G$, consider the set $o(M \cap E_1)$. For any vertex $v \in V$ such that $((v, 0), (v, 1))$ is not in $M$, both $(v, 0)$ and $(v, 1)$ have distinct incident edges in $M \cap E_1$, say $e_0$ and $e_1$. $o(e_0)$ is an outgoing edge from $v$, and $o(e_1)$ is an incoming edge. No other edges in $o(M \cap E_1)$ are incident upon $v$. Thus, all vertices of $E$ that have incident edges in $o(M \cap E_1)$ have both in-degree and out-degree exactly $1$. Thus, $o(M \cap E_1)$ is the disjoint union of cycles in $G$. Let these cycles be $C_1, \dots, C_m$. Since $c \in V$ does not have the $((c, 0), (c, 1))$ edges in $G'$, $c$ is necessarily a part of one of these cycles, say it's $C_1$. Then $w'(M) = \sum_{i=1}^m |C_i| + \sum_{i=1}^m W(C_i) \cdot (n+1) \geq \sum_{i=1}^m |C_i| + W(C_1) \cdot (n+1) \geq \sum_{i=2}^m |C_i| + W(C^*) \cdot (n+1) + |C^*|$. If $m > 1$, then this is $> W(C^*) \cdot (n+1) + |C^*|$, implying that $M$ is not a minimum weight perfect matching of $G$. Thus, $m = 1$. This implies that $o(M \cap E_1)$ is necessarily just a simple cycle containing $s$, say $C$, and $M$ has weight $W(C) \cdot (n+1) + |C|$. Note that for any cycle $C$ such that $W(C) \cdot (n+1) + |C| \leq W(C^*) \cdot (n+1) + |C^*|$, since $|C^*| \leq n$  and $W(C^*) \leq W(C)$, we have $W(C^*) = W(C)$ and $|C| = |C^*|$. Thus, $o(M \cap E_1)$ is necessarily the minimum weight minimum hop-length cycle of $G$ containing $c$.\\

\noindent
The algorithm is now straightforward: We first compute $G'$ - this is easy to do in $\CL$. We then compute the minimum weight perfect matching $M$ of $G'$ using \cite{AgarwalaMertz25}. Finally, we compute $o(M \cap E_1)$ and verify that this is a simple cycle of $G$. If not, $G$ does not contain a simple cycle containing $c$. Else, this is the minimum weight minimum hop-length simple cycle of $G$ containing $s$, and we can output it. \\

\noindent
We note a small technicality: the output of the matching algorithm of \cite{AgarwalaMertz25} may differ depending on the initial state of the catalytic tape. In order to ensure stability, we can simply partition the catalytic tape before hand into a separate section used only for calls to the matching algorithm. Since the matching algorithm is never used recursively, this ensures that each call to the matching algorithm with the same graph indeed results in the same matching. 
\end{proof}

\begin{lemma*}[\Cref{size-k-common-independent-set-reduces-to-perfect}]
There exists a $\Logspace$ machine which:
\begin{enumerate}
    \item Takes as input weighted linear matroids $M_1 = (S, \mathcal{I}_1),\ M_2 = (S, \mathcal{I}_2),\ W\colon S \rightarrow Z^{\leq \poly(n)}$ with linear representations $L_1$ and $L_2$.
    \item Computes weighted linear matroids $M_1' = (S \cup S', \mathcal{I}'_1),\ M_2' = (S \cup S', \mathcal{I}'_2),\ W':S \cup S' \rightarrow Z^{\leq \poly(n)}$ such that if $W$ isolates a size $k$ common independent set $I_k$ of $M_1$ and $M_2$, then $W'$ isolates a perfect common independent set $I$ of $M_1'$ and $M_2'$ which satisfies $I \cap S = I_k$.
\end{enumerate}
\end{lemma*}

\begin{proof}
We will assume without loss of generality that $L_1$ and $L_2$ are both $m\times n$ matrices over $\mathbb{F}$. If $L_1$ and $L_2$ have a different number of rows, note that we can pad the shorter matrix with all $0$s rows in order to achieve equality.\\

\noindent
This proof is technical, but largely follows from very basic properties of matrices. Thus, we will opt for succinctness at various points of this proof.\\

\noindent
For $1\leq i\leq m$ and $1 \leq j \leq m-k$, let $e^{(i)} \in \mathbb{F}^{m}$ be the vector such that $e^{(i)}_\ell=1$, if $\ell=i$, and $0$ otherwise, and let $f^{(j)} \in \mathbb{F}^{m-k}$ be the vector such that $f^{(j)}_{\ell} = 1$, if $\ell = j$, and $0$ otherwise. We define $C$ and $D$ to be $m\times m(m-k)$ and $(m-k)\times m(m-k)$ dimension matrices respectively as follows: for $1\leq i\leq m$ and $1 \leq j\leq m-k$, let $s_{ij}=(i-1)(m-k)+j$. The $s_{ij}^{th}$ columns $C_{s_{ij}}$ and $D_{s_{ij}}$ are equal to $e^{(i)}$ and $f^{(j)}$ respectively.\\

\noindent
Now, consider the following linear representations $L_1'$ and $L_2'$ of matroids $M_1'$ and $M_2'$:\\

    \begin{align*}
        L'_1&=\begin{pmatrix}
        L_1 & C & 0_{m \times m(m-k)}\\
        0_{(m-k) \times n}
    & 0_{(m-k) \times  m(m-k)} & D
    \end{pmatrix}
    \\
    L'_2&=\begin{pmatrix}
        L_2 & 0_{m \times m(m-k)} & C\\
        0_{(m-k) \times n}
    & D & 0_{(m-k) \times  m(m-k)},
\end{pmatrix}
    \end{align*}

\ \\

\noindent
Let $M_1' = (S \cup A \cup B, \mathcal{I}_1')$ and $M_2' = (S \cup A \cup B, \mathcal{I}_2')$, where $A$ corresponds to the set of columns which are $C$ in $L'_1$, and $B$ to those which are $C$ in $L'_2$.\\

\noindent
The key insight is that any common independent set $I \in \mathcal{I}_1 \cap \mathcal{I}_2$ of size $\ell$, can be extended to a common independent set $I'$ of size $\min\{ 2m- \ell,  \ell +2m-2k\}$ such that $I' \in \mathcal{I}'_1 \cap \mathcal{I}'_2$. Indeed, observe that the unit vectors $\{e^{(i)} \ | \ i \in [m]\}$ in $C$ can extend any linearly independent set of columns $I \subseteq L_1$ to a size $m$ linearly independent set, and the unit vectors $\{f^{(i)} \ | \ i \in [m-k]\}$ in $D$ are linearly independent and supported on a different set of rows. Moreover, $C$ and $D$ consists of all the different combinations of unit vectors $(e^{(i)}, f^{(j)})$. Thus, we can augment the set $I$ with a set of at least $\min\{m-k, m-l\}$ elements from $A$, and the same from $B$.\\

\noindent
Also, note that since $C$ is supported on the same set of $m$ rows as $L_1$ and $L_2$, we can take at most $m-l$ elements each from $A$ and $B$, and since $D$ is supported on $m-k$ rows, we can take at most $m-k$ elements each from $A$ and $B$. This together shows that $\max_{I' \in \mathcal{I}'_1 \cap \mathcal{I}'_2 : I \subseteq I'} |I'| = \min\{2m - 2k + l, 2m - l\}$. This is maximised only when $|I| = k$, in which case the independent set $I'$ has size $2m - k$, and is thus perfect. \\

\noindent
This shows that any maximum sized common independent set $I$ of $M_1'$ and $M_2'$ is perfect in $L_1'$ and $L_2'$, and satisfies $|I \cap S| = k$. Moreover, for any size $k$ common independent set $I$ of $M_1$ and $M_2$, there exists a maximum sized common independent set of $M_1'$ and $M_2'$ containing $I$. The only remaining task is to assign weights to the elements of $S \cup A \cup B$ such that the uniqueness of a size $k$ common independent set of $M_1$ and $M_2$ implies the uniqueness of a minimum weight perfect independent set of $M_1'$ and $M_2'$. \\

\noindent
We define the weight $W' \colon S \cup A \cup B \rightarrow \mathbb{Z}^{\leq \poly(n)}$ as follows:
\begin{enumerate}
    \item 
    If $e\in S$, we have $W'(e)=10 m^4 \cdot W(e)$.
    \item 
    Else, if $e \in A \cup B$, let $s_{ij}$ be the column of $C$ associated to $e$. We define $W'(e)=ij$.
\end{enumerate}

\noindent
Our goal is to prove that $W'$ has the following two properties:
\begin{enumerate}
    \item 
    For any minimum weight perfect common independent set $I \in \mathcal{I}'_1 \cap \mathcal{I}'_2$, $I\cap S = I_k$, and
    \item 
    This minimum weight perfect common independent set $I$ is unique.
\end{enumerate}

\noindent
For the first property, let $T_1,T_2 \in \mathcal{I}'_1\cap \mathcal{I}'_2$ be common independent sets such that $W(T_1\cap S) < W(T_2\cap S)$. We have:
\begin{align*}
    W'(T_2) - W'(T_1)=& W'(T_2\cap S)-W'(T_1\cap S)\\
    &+W'(T_2 \setminus S)-W'(T_1 \setminus  S)
    \\
    \geq& 10 m^4\left(W(T_2\cap S)-W(T_1\cap S)\right) - W'(T_1 \setminus  S)
    \\
    \geq& 10 m^4 -2m^2(m-k)^2 > 0.
    \\
    \implies W'(T_2) > W'(T_1)
\end{align*}

\noindent
Thus, any minimum weight perfect common independent set $I \in \mathcal{I}'_1 \cap \mathcal{I}'_2$ necessarily satisfies $I \cap S = I_k$.\\

\noindent
Now, we will prove that the choice of $S_1^* \subseteq A$ and $S_2^* \subseteq B$ such that $I_k \cup S_1^* \cup S_2^*$ is a minimum weight maximum common independent set of $M_1^*$ and $M_2^*$ is unique. In particular, we will prove this for $S_1^* \subseteq A$, and the same proof will work for $S_2^* \subseteq B$. \\

\noindent
Let $L_1(I_k)$ refer to the set of columns corresponding to the elements of $I_k$ in the matrix $L_1$. Let any $E \subseteq \{e^{(\ell)} \ | \ \ell \in [m]\}$ such that $|E| = m-k$ and $L_1(I_k) \cup E$ is linearly independent over $\mathbb{F}^m$ be referred to as an "extension" of $I_k$. There exists a unique extension $E^* = \{e^{(i_1)}, \dots, e^{(i_{m-k})}\}$ where $i_1 < i_2 < \dots < i_{m-k}$ such that for any other extension $E' = \{e^{(j_1)}, \dots, e^{(j_{m-k})}\}$ we have $\forall l \in [m-k], \ i_l \leq j_l$ and $\exists l$ such that $i_l < j_l$. \footnote{This is a well known fact about matroids and is central to their study in greedy algorithms. More information about it can be found in chapter 40 of \cite{Schrijver03Textbook}.}.\\

\noindent
Now, consider any minimum weight maximum common independent set $I$ of $M_1^*$ and $M_2^*$. $I \cap A$ in the first $m$ rows of the matrix $L_1'$ consists of columns $\{s_{i_1j_1}, \dots, s_{i_{m-k} j_{m-k}}\}$ of $C$. Here, the set $E = \{e^{(i_1)}, \dots, e^{(i_{m-k})}\}$ is an extension of $I_k$, $J = \{j_1, \dots, j_{m-k}\}$ is just a permutation of elements of $[m-k]$, and $W'(I \cap A) = \sum_{l=1}^{m-k} i_l \cdot j_l$. Additionally, for any extension $E'$, and any permutation $J'$ of $[m-k]$, the set $S_1 \subseteq A$ corresponding to $E'$ and $J'$ satisfies $S_1 \cup I_k \in \mathcal{I}_1' \cap \mathcal{I}'_2$.  Now, note that for any set $S_1$ built from an extension $E$ and permutation $J$ of $[m-k]$, the set $S_1'$ built from the extension $E'$ and permutation $J$ has strictly smaller weight. This implies that $E$ is necessarily $E^*$ when $I$ is a minimum weight maximum common independent set. All we have to show now is that the permutation $J$ is also unique. This follows directly from Claim 3.4 in \cite{AgarwalaMertz25}. The claim in \cite{AgarwalaMertz25} is stated in terms of perfect matchings of a complete bipartite graph, but it is easy to observe that this is exactly the same as our permutations when the graph is built on $E^* \times [m-k]$ with the weight of edge $(i_k, j) = i_k \cdot j$.\\

\noindent
This completes the proof. We can simply output $L_1', L_2'$, and $W'$.
\end{proof}